\documentclass[a4paper,12pt]{article}

\usepackage{latexsym}
\usepackage{amsmath}
\usepackage{amssymb}
\usepackage{amsfonts}
\usepackage{amsthm}
\usepackage[latin1]{inputenc}
\usepackage{graphicx}
\usepackage{float,times}
\usepackage[usenames]{color}
\usepackage{longtable}
\usepackage{rotating}
\usepackage[round]{natbib}
\usepackage{fancyhdr}
\usepackage{multirow}
\usepackage[margin=2cm]{geometry}

\newtheoremstyle{example}{\topsep}{\topsep}%
     {}
     {}
     {\bfseries}
     {}
     {\newline}
     {\thmname{#1}\thmnumber{ #2}\thmnote{ #3}}

\theoremstyle{example}
\theoremstyle{theorem}
\newtheorem{theorem}{Theorem}
\theoremstyle{theorem}
\theoremstyle{proposition}
\theoremstyle{corollary}

\def\keywords{\vspace{.5cm}
{\textit{Keywords}:\,\relax%
}}

\title{On posterior propriety for the Student-$t$ linear regression model under Jeffreys priors}

\author{Catalina A. Vallejos  and   Mark F.J. Steel\\
{\small Department of Statistics,
University of Warwick, Coventry, CV4 7AL, U.K.}\\
{\small C.A.Vallejos-Meneses@warwick.ac.uk and M.Steel@warwick.ac.uk}}

\date{}

\begin{document}

\maketitle

\begin{abstract}
Regression models with fat-tailed error terms are an increasingly popular choice to obtain more robust inference to the presence of outlying observations. This article focuses on Bayesian inference for the Student-$t$ linear regression model under objective priors that are based on the Jeffreys rule. Posterior propriety results presented in \cite{fonsecaetal2008} are revisited and one result is corrected. In particular, it is shown that the standard Jeffreys-rule prior precludes the existence of a proper posterior distribution. 
\end{abstract}


\section{Introduction} \label{SectionIntroduction}

The normal assumption in linear regression models does not always provide an appropriate fit to real datasets. Data often require more flexible errors, capable of accommodating outlying observations. A popular alternative is to assume a Student-$t$ distribution for the error term \citep[see for example][]{west1984,langeetal1989,fs1999,fonsecaetal2008}. The choice of a prior is very challenging when conducting Bayesian inference under Student-$t$ sampling. While some ``standard'' priors can be adopted for the regression and scale parameters, there is no consensus about a prior distribution for the degrees of freedom ($\nu$). \cite{villawalker2013} provide a comprehensive discussion of the literature. The seminal paper by \cite{fonsecaetal2008} is, as far as we know, the first attempt to base an objective prior for $\nu$ on formal rules and introduces two objective priors based on the Jeffreys rule. They propose the original Jeffreys-rule prior and one of its variants, the independence Jeffreys prior (which treats the regression parameters independently). These priors have been considered in several subsequent articles. \cite{ho2012} and \cite{villawalker2013} used both priors. In the context of skew-$t$ models, the independence Jeffreys prior was  used in \cite{juarezsteel2010} and \cite{brancoetal2012}.

This note is a follow-up of \cite{fonsecaetal2008}. Their posterior propriety results are revisited and one of their results is corrected. In particular, it is shown that the prior based on the original Jeffreys rule precludes the existence of a proper posterior distribution. Nevertheless, the independence Jeffreys prior yields a well-defined posterior distribution.

The Student-$t$ linear regression model is presented in Section \ref{SectionRegression}, which also includes the priors presented in \cite{fonsecaetal2008}. 
Posterior propriety under these priors is examined in Section \ref{SectionPropriety}, while Section \ref{SectionConclusion} concludes.

\section{Bayesian Student-$t$ linear regression model} \label{SectionRegression}

Let $Y=(Y_1,\ldots,Y_n)' \in \mathbb{R}^n$ represent $n$ independent random variables generated by the linear regression model \begin{equation} \label{eqLinearRegression} Y_i = x'_i \beta+\sigma\epsilon_i, \hspace{1cm} i=1,\ldots,n, \end{equation} where $x_i$ is a vector containing the value of $p$ covariates associated with observation $i$, $\beta \in \mathbb{R}^p$ is a vector of regression parameters and $\epsilon_i$ has Student-$t$ distribution with mean zero, unitary scale and $\nu$ degrees of freedom. The Bayesian model is completed using Jeffreys priors, which require the Fisher information matrix. Similarly to \cite{fonsecaetal2008} (they parameterize with respect to $\sigma$ instead), the Fisher information matrix for the model in (\ref{eqLinearRegression}) is given by \begin{small}\begin{equation} I(\beta,\sigma^2,\nu)=
\left(%
\begin{array}{ccc}
  \frac{1}{\sigma^2}\frac{\nu+1}{\nu+3}\sum_{i=1}^n x_i x'_i & 0 & 0 \\
  0 & \frac{n}{2\sigma^4}\frac{\nu}{\nu+3} & -\frac{n}{\sigma^2} \frac{1}{(\nu+1)(\nu+3)} \\
  0 & -\frac{n}{\sigma^2} \frac{1}{(\nu+1)(\nu+3)} & \frac{n}{4} \big[\Psi'(\frac{\nu}{2})-\Psi'(\frac{\nu+1}{2})-\frac{2(\nu+5)}{\nu(\nu+1)(\nu+3)}\big]  \\
\end{array}%
\right),
\end{equation} \end{small} where $\Psi'(\cdot)$ denotes the trigamma function. Hence, the Jeffreys-rule and the independence Jeffreys (which deals separately with the blocks for $\beta$ and $(\sigma^2,\nu)$) 
priors are respectively given by \begin{small}\begin{eqnarray} \pi^J(\beta,\sigma^2,\nu) & \propto &   \frac{1}{(\sigma^2)^{1+p/2}}\left(\frac{\nu+1}{\nu+3}\right)^{p/2} \sqrt{\frac{\nu}{\nu+3}} \sqrt{\Psi'\left(\frac{\nu}{2}\right)-\Psi'\left(\frac{\nu+1}{2}\right) -\frac{2(\nu+3)}{\nu(\nu+1)^2}}, \\
\pi^{I}(\beta,\sigma^2,\nu) & \propto & \frac{1}{\sigma^2} \sqrt{\frac{\nu}{\nu+3}} \sqrt{\Psi'\left(\frac{\nu}{2}\right)-\Psi'\left(\frac{\nu+1}{2}\right) -\frac{2(\nu+3)}{\nu(\nu+1)^2}}. \end{eqnarray} \end{small} These priors have been proposed in \cite{fonsecaetal2008} and can be written as \begin{equation} \label{eqPrior} \pi(\beta,\sigma^2,\nu) \propto \frac{1}{(\sigma^2)^a} \pi(\nu), \end{equation} where $\pi(\nu)$ is the component of the prior that depends on $\nu$, $a=1+p/2$ for the Jeffreys-rule prior and $a=1$ for the independence Jeffreys prior. As shown in \cite{fonsecaetal2008}, $\pi(\nu)$ is a proper density function of $\nu$ for both priors. 

\section{Posterior propriety} \label{SectionPropriety}

Verifying the existence of the posterior distribution is mandatory under the prior in (\ref{eqPrior}), which is not a proper probability density function of $(\beta,\sigma^2,\nu)$. Corollary 2 in \cite{fonsecaetal2008} states that, provided $n>p$, the posterior distribution is well-defined under the Jeffreys-rule and the independence Jeffreys priors. Their proof refers to Theorem 1 in \cite{fs1999}, but unfortunately this theorem does not cover the Jeffreys-rule prior, as it assumes that $a=1$ in (\ref{eqPrior}). A necessary condition for the existence of the posterior distribution is now  provided in the following Theorem. 


\begin{theorem} \label{theoProprietyLST2}
Let $y=(y_1,\ldots,y_n)'$ be $n$ independent observations from model (\ref{eqLinearRegression}). Define $X=(x_1, \ldots, x_{n})'$ and assume that $n>p$ and the rank of $X$ is $p$. Under the prior in (\ref{eqPrior}), a necessary condition for posterior propriety is $\pi(\nu)=0$ for all $\nu \in \left(0,\frac{2a-2}{n-p}\right]$. \end{theorem}

\begin{proof}
Define $D=\mbox{diag}(\lambda_1,\ldots,\lambda_{n})$. After some algebraic manipulation, \begin{small}
\begin{equation} \label{appendixInequalityCensored} f_{Y}(y) = \int_{\mathbb{R}_+} \int_{\mathbb{R}_+} \int_{\mathbb{R}^p} \int_{\mathbb{R}_+^{n}} \frac{\prod_{i=1}^{n} \lambda_i^{\frac{1}{2}}}{ (2\pi \sigma^2)^{\frac{n}{2}}}\,\frac{e^{-\frac{1}{2\sigma^2}\left[(\beta-b)'A(\beta-b)+S^2(D,y)\right]}}{{\sqrt{\det(X'D X)}}} \frac{\pi(\nu)}{(\sigma^2)^a}\left[ \prod_{i=1}^{n} f^G_{\Lambda_i}(\lambda_i|\nu) \,d \lambda_i \right] \,d \beta \,d \sigma^2 \,d \nu, \end{equation}
\end{small}
where $A=X'D X$, $b=A^{-1} X' D y$ and $S^2(D,y)=y'D y - y' D X (X'D X)^{-1} X' D y$. Using Fubini's theorem and integrating first with respect to $\beta$, we have
\begin{equation}f_{Y}(y) \propto \int_{\mathbb{R}_+} \int_{\mathbb{R}_+} \int_{\mathbb{R}_+^{n}} (\sigma^2)^{-\frac{n+2a-p}{2}} \frac{\prod_{i=1}^{n} \lambda_i^{\frac{1}{2}}}{\sqrt{\det(X'D X)}} \,e^{-\frac{S^2(D,y)}{2\sigma^2}} \pi(\nu) \left[\prod_{i=1}^{n}  f^G_{\Lambda_i}(\lambda_i|\nu) \,d \lambda_i \right] \,d \nu \,d \sigma^2.
\end{equation}
After integrating with respect to $\sigma^2$, it follows that
\begin{equation} \label{appProperCondition} f_{Y}(y) \propto  \int_{\mathbb{R}_+} \int_{\mathbb{R}_+^n} \prod_{i=1}^{n} \lambda_i^{\frac{1}{2}} [\det(X'D X)]^{-\frac{1}{2}} [S^2(D,y)]^{-\frac{n+2a-p-2}{2}} \pi(\nu) \left[ \prod_{i=1}^{n}  f^G_{\Lambda_i}(\lambda_i|\nu) \,d \lambda_i \right] \,d \nu, \end{equation}
as long as $n+2a-p-2>0$ and $S^2(D,y)>0$. Since $n>p$ we know that $S^2(D,y)>0$ a.s. Analogously to Lemma 1 in \cite{fs1999}, $f_{Y}(y)$ has upper and lower bounds proportional to
\begin{equation} \label{appProperCondition1} \int_{\mathbb{R}_+} \int_{0<\lambda_1<\cdots<\lambda_{n} <\infty} \prod_{i \notin \{m_1,\ldots,m_p\}} \lambda_i^{\frac{1}{2}} \lambda_{m_{p+1}}^{-\frac{n+2a-p-2}{2}} \pi(\nu) \left[\prod_{i=1}^{n} f^G_{\Lambda_i}(\lambda_i|\nu) \,d \lambda_i \right] \,d \nu, \end{equation} where \begin{small}\begin{eqnarray*} \prod_{i=1}^p \lambda_{m_i} & \equiv & \max\left\{\prod_{i=1}^p \lambda_{l_i}: \det\left(x_{l_1} \cdots x_{l_p}\right) \neq 0, l_1, \ldots l_p \in \{1,\ldots,n\}\right\},\\
\prod_{i=1}^{p+1} \lambda_{m_i} & \equiv & \max\left\{\prod_{i=1}^{p+1} \lambda_{l_i}: \det\left(%
\begin{array}{ccc}
  x_{l_1} & \cdots & x_{l_{p+1}} \\
  y_{l_1} & \cdots & y_{l_{p+1}} \\
\end{array}%
\right) \neq 0, l_1, \ldots l_p \in \{1,\ldots,n\} \right\}. \end{eqnarray*}\end{small}
Upper and lower bounds for the previous integral can be found using the inequality \citep{fs1999,fs2000} \begin{equation} \label{appendixInequality} \frac{\lambda_{i+1}^v}{v} \,e^{-r\lambda_{i+1}} \leq
\int_0^{\lambda_{i+1}} \lambda_i^{v-1} \,e^{-r\lambda_i} \,d
\lambda_i \leq \frac{\lambda_{i+1}^v}{v}, \hspace{1cm} r,v>0.
\end{equation} The integral in (\ref{appendixInequality}) is not finite for $v\le 0$. Barring a set of zero Lebesgue measure, $\lambda_{m_{p+1}}=\lambda_{(n-p)}$, where $\lambda_{(n-p)}$ is the $(n-p)$-th order statistic of $\lambda_1,\ldots,\lambda_{n}$. After integrating with respect to the $n-p-1$ smallest $\lambda_i$'s, (\ref{appProperCondition1}) has a lower bound proportional to
\begin{small}\begin{equation} \label{appendixLower}
\int_0^{\infty}\int_{\Lambda^*}\left[\frac{\left(\frac{\nu}{2}\right)^{\frac{\nu}{2}}}{\Gamma\left(\frac{\nu}{2}\right)}\right]^{n-p}
\frac{\left[\frac{\nu+1}{2}\right]^{-(n-p-1)}}{(n-p-1)!}
\lambda_{(n-p)}^{c-1} \,e^{-\frac{(n-p)\nu}{2}\lambda_{(n-p)}} \,d \lambda_{(n-p)} \left[\prod_{i=n-p+1}^{n}  f^G_{\Lambda_i}(\lambda_{(i)}|\nu) \,d \lambda_{(i)} \right] \pi(\nu) \,d \nu,
\end{equation} \end{small}
where $\Lambda^*=\{(\lambda_{(n-p)},\ldots,\lambda_{(n)}):0<\lambda_{(n-p)}<\cdots<\lambda_{(n)}<\infty\}$ and $c=-\frac{n+2a-p-3}{2}+\frac{\nu}{2}+\frac{(n-p-1)(\nu+1)}{2}=\frac{\nu(n-p)+2-2a}{2}$. When integrating with respect to $\lambda_{(n-p)}$ we need  $c>0$ in order to have a finite integral in (\ref{appendixLower}). Hence, the propriety of the posterior distribution requires $\nu > \frac{2a-2}{n-p}$.
\end{proof}

As a consequence, the posterior distribution of $(\beta,\sigma^2,\nu)$ is not proper if $a>1$ and the range of $\nu$ is $(0,\infty)$. In particular, the Jeffreys-rule prior (for which $a=1+p/2$) does not lead to a proper posterior distribution and Bayesian inference is thus precluded with this prior. The independence Jeffreys prior satisfies the necessary condition in Theorem \ref{theoProprietyLST2}, but this does not guarantee posterior existence. Nevertheless, posterior propriety under the independence Jeffreys prior for $n>p$ is ensured by Theorem 1 in \cite{fs1999}, as indicated in \cite{fonsecaetal2008}.

\section{Concluding remarks} \label{SectionConclusion}

The choice of a prior distribution for the degrees of freedom under Student-$t$ sampling is a very challenging task. \cite{fonsecaetal2008} adopt Jeffreys principles to find objective priors for $\nu$. This is an important addition to the previous literature in which much more ad-hoc priors were  used ({\it e.g.}~the exponential prior in Geweke, 1993 and Fern\'andez and Steel, 1999)\nocite{geweke1993}. Here we show that the Jeffreys-rule prior does not produce a proper posterior distribution, in contrast to the claim in \cite{fonsecaetal2008}. We believe it is crucial to point this out to the scientific community to avoid  meaningless inference and misleading conclusions. The Jeffreys-rule prior under Student-$t$ sampling has also been used in \cite{ho2012} and \cite{villawalker2013}. For this prior, \cite{fonsecaetal2008} and \cite{villawalker2013} observe very poor frequentist coverage of the 95\% credible intervals for $\nu$ when the sample size is small ($n=30$). 
For small sample size the lower bound required on the support of $\nu$ (here equal to $p/(n-p)$) may easily be violated by samples from the posterior, so this poor empirical performance might be linked to the impropriety shown here. 
Posterior propriety can be verified under the independence Jeffreys prior, and its use as an objective prior is recommended to practitioners.

\subsection*{Acknowledgements}
Catalina Vallejos acknowledges research
support from the University of Warwick and from the Department of Statistics of the Pontificia Universidad Católica de Chile. We thank Marco Ferreira for his constructive comments.

\bibliography{SMLNpaper}
\bibliographystyle{Chicago}

\end{document}